\documentclass[12pt,draftclsnofoot,journal,onecolumn]{IEEEtran}
\usepackage{mathtools}
\usepackage{ifpdf}
\usepackage{amsthm}
\usepackage{amsmath}
\usepackage{nicefrac}
\usepackage{cite}
\usepackage{amsfonts}
\usepackage{comment}
 \usepackage{url}
 \usepackage{amssymb}
 \usepackage[paper=letterpaper,margin=0.75in]{geometry}
 \usepackage[ansinew]{inputenc}

\usepackage{lipsum}
\usepackage{hyperref}

\newtheorem{thm}{Theorem}
\newtheorem{lem}{Lemma}
\newtheorem{fact}{Fact}
\theoremstyle{definition}

\makeatletter
\renewenvironment{proof}[1][\proofname] {\par\pushQED{\qed}\normalfont\topsep6\p@\@plus6\p@\relax\trivlist\item[\hskip\labelsep\bfseries#1\@addpunct{.}]\ignorespaces}{\popQED\endtrivlist\@endpefalse}
\makeatother

\begin{document}
\title{Covert Communications on Poisson Packet Channels}

\author{
   \IEEEauthorblockN{Ramin Soltani\IEEEauthorrefmark{1},
            Dennis Goeckel\IEEEauthorrefmark{1}, Don Towsley\IEEEauthorrefmark{2}, and Amir Houmansadr\IEEEauthorrefmark{2}}

\IEEEauthorblockA{\IEEEauthorrefmark{1}Electrical~and~Computer~Engineering~Department,~University~of~Massachusetts,~Amherst,
    \{soltani, goeckel\}@ecs.umass.edu\\}
    \IEEEauthorblockA{\IEEEauthorrefmark{2}College of Information and Computer Sciences, University of Massachusetts, Amherst,
    \{towsley, amir\}@cs.umass.edu}
        
                       \thanks{ This work has been supported by the National Science Foundation under grants ECCS-1309573 and CNS-1525642.}

\thanks{ This work has been presented at the 53rd Annual Allerton Conference on Communication, Control, and Computing,  October 2015. The extension of this work to Renewal packet channels is~\cite{soltani2016allertonarxiv}, which has been presented at the 54th Annual Allerton Conference on Communication, Control, and Computing,  October 2016.}
\thanks{Personal use of this material is permitted. Permission from IEEE must be obtained for all other uses, in any current or future media, including reprinting/republishing this material for advertising or promotional purposes, creating new collective works, for resale or redistribution to servers or lists, or reuse of any copyrighted component of this work in other works. DOI: \href{https://doi.org/10.1109/ALLERTON.2015.7447124}{10.1109/ALLERTON.2015.7447124}}

}

\date{}
\maketitle
\thispagestyle{plain}
\pagestyle{plain}

\begin{abstract}
Consider a channel where authorized transmitter Jack sends packets to authorized receiver Steve according to a Poisson process with rate $\lambda$ packets per second for a time period $T$. Suppose that covert transmitter Alice wishes to communicate information to covert receiver Bob on the same channel without being detected by a watchful adversary Willie. We consider two scenarios. In the first scenario, we assume that warden Willie cannot look at packet contents but rather can only observe packet timings, and Alice must send information by inserting her own packets into the channel. We show that the number of packets that Alice can covertly transmit to Bob is on the order of the square root of the number of packets that Jack transmits to Steve; conversely, if Alice transmits more than that, she will be detected by Willie with high probability. In the second scenario, we assume that Willie can look at packet contents but that Alice can communicate across an $M/M/1$ queue to Bob by altering the timings of the packets going from Jack to Steve. First, Alice builds a codebook, with each codeword consisting of a sequence of packet timings to be employed for conveying the information associated with that codeword.  However, to successfully employ this codebook, Alice must always have a packet to send at the appropriate time.  Hence, leveraging our result from the first scenario, we propose a construction where Alice covertly slows down the packet stream so as to buffer packets to use during a succeeding codeword transmission phase.  Using this approach, Alice can covertly and reliably transmit $\mathcal{O}(\lambda T)$ covert bits to Bob in time period $T$ over an $M/M/1$ queue with service rate $\mu > \lambda$.
\end{abstract}



\section{Introduction}

The provisioning of security and privacy is critical in modern communication systems.  The vast majority of research in this area has focused on protecting the message content through encryption \cite{talb2006} or emerging methods in information-theoretic security \cite{bloch11pls}.  However, there are applications where the very existence of the message must be hidden from a watchful adversary.  For example, in the organization of social unrest against an authoritarian regime, the mere existence of an encrypted message would likely lead to the shut down of that communication and possible punishment of the users.  Other applications include hiding the presence or scale of military operations, or removing the ability of users to be tracked in their everyday activities.

Whereas the provisioning of undetectable communications in objects with symbols drawn from a finite field and transmitted over noiseless channels (e.g. images) has been extensively studied in the steganogaphic community over a number of years \cite{ker07pool},
the consideration of hiding messages on continuous-valued noisy channels has only recently drawn considerable attention.  In particular, despite the extensive historical use of spread spectrum systems 
for low probability of detection (LPD) communications, the fundamental limits of covert communications over the additive white Gaussian noise (AWGN) channel were only recently obtained \cite{bash12sqrtlawisit,bash_jsac2013}.  The work of \cite{bash12sqrtlawisit, bash_jsac2013} motivated significant further work (e.g. \cite{boulat_commmagg,commL,bash13quantumlpdisit,
bash14timing,
soltani14netlpdallerton,
tammy_asilomar2015, che13sqrtlawbscisit,kadhe14sqrtlawmultipathisit,hou14isit,bloch15covert-isit,wang15covert-isit,lee14posratecovert,
jaggi_uncertain}), with these works taken together providing an almost complete characterization of the foundational limits of covert communications over discrete memoryless channels (DMCs) and AWGN channels.  

In this paper, we turn our attention to the consideration of a covert timing channel appropriate for packet-based communication channels, e.g., computer networks.  More specifically, we study how packet insertion or inter-packet delays can be used for covert communications. The use of inter-packet delays for covert communication was first explored by Girling~\cite{Girling87} and was later studied by a number of other authors~\cite{berk2005detection,shah2006keyboards,liu2009hide,houmansadr2011coco}.  In particular, Anantharam et~al.~\cite{verdubitsq} derived the Shannon capacity of the timing channel with a single-server queue, and Dunn et~al.~\cite{dunn2009} analyzed the secrecy capacity of such a system.  Considerable work has focused on quantifying and optimizing the capacity of timing channels~\cite{millen1989finite,Anand1998inf,sekke2007capacity,sekke2009,Mosko92Capac,Moskowitz1994} by leveraging information theoretic analysis and the use of various coding techniques \cite{kiyavash2009J,kiyavash2009A,Arch2012}. Moreover, methods for detecting timing channels~\cite{gianvecchio2007detecting} as well as eluding the detection by leveraging the statistical properties of the legitimate channel have been proposed~\cite{gianvecchio2008model}. 

In this paper, we consider a channel where an authorized (overt) transmitter, Jack, sends packets to an authorized (overt) receiver, Steve, within a time period $T$, where the timings of packet transmission are modeled by a Poisson point process with rate $\lambda$ packets per second.  Covert transmitter Alice may wish to transmit data to a covert receiver, Bob, on this channel in the presence of a warden, Willie, who is monitoring the channel between Alice and Bob precisely to detect such transmissions.  We consider two scenarios in detail. In the first scenario, we assume:  (1)  the warden Willie is not able to see packet contents, and therefore cannot verify the authenticity of the packets (e.g., whether they are actually sent by Jack); and  (2) Alice is restricted to packet insertion.  Willie is aware that the timing of the packets of the overt communication link follows a Poisson point process with rate $\lambda$, so he seeks to apply hypothesis testing to verify whether the packet process has the proper characteristics.   We show that (see Theorem~1) if Alice decides to transmit $\mathcal{O}\left(\sqrt{\lambda T}\right)$ packets to Bob, she can keep Willie’s sum of error probabilities $\mathbb{P}_{FA} + \mathbb{P}_{MD} > 1-\epsilon$ for any $0<\epsilon<1$, where $\mathbb{P}_{FA}$ and $\mathbb{P}_{MD}$ are Willie's probability of false alarm and missed detection, respectively. Conversely, we prove that if Alice transmits ${\omega}\left(\sqrt{\lambda T}\right)$ packets, she will be detected by Willie with high probability. 

In the second scenario, we assume that Willie can look at packet contents and therefore can verify packets' authenticity.  Thus, Alice is not able to insert packets, but we allow Alice in this scenario to alter the packet timings to convey information to Bob, whom is receiving the packets through a $M/M/1$ queue with service rate $\mu>  \lambda$.   To do such, Alice designs an efficient code \cite{verdubitsq}, where a codeword consists of a sequence of packet timings drawn from the same process as the overt traffic; hence, a codeword transmitted with those packet timings is undetectable.  However, there is a causality constraint, as Alice clearly cannot send the next packet (i.e. codeword symbol) unless she has a packet from the Jack to Steve link available to transmit.   This suggests the following two-stage process.  In the first stage, Alice covertly slows down the transmission of the packets from Jack to Steve so as to buffer some number of packets.  In the second stage, Alice continues to add packets transmitted by Jack to her buffer while releasing packets with the inter-packet delay appropriate for the chosen codeword.   Alice's scheme breaks down when her buffer is empty at any point before completing the codeword transmission.   Hence, the question becomes:  how long must Alice collect packets during the first stage so as to guarantee (with high probability) that she will not run out of packets before the completion of codeword transmission during the second stage?

First in Lemma~1, we show that Alice can collect and store $\mathcal{O}\left(\sqrt{\lambda T}\right)$ packets in an interval of length $T$ while keeping Willie’s sum of error probabilities $\mathbb{P}_{FA} + \mathbb{P}_{MD} > 1-\epsilon$ for any $0<\epsilon<1$; conversely, if she collects more, she will be detected by Willie with high probability. Building on Lemma~1, we prove (Theorem~2) that, using our two-stage covert communications approach, Alice can reliably transmit $\mathcal{O}(\lambda T)$ covert bits to Bob within the time period $T$ while keeping Willie’s sum of error probabilities bounded as $\mathbb{P}_{FA} + \mathbb{P}_{MD} > 1-\epsilon$ for any $0<\epsilon<1$. 

The remainder of the paper is organized as follows. In Section~\ref{modcons}, we present the system model and the metrics employed for the two scenarios of interest. Then, we provide constructions and their analysis for the two covert communication scenarios in Sections~\ref{theorem1} and \ref{theorem2}. 

\section{System Model and Metrics}
\label{modcons}

\subsection{System Model}
Suppose that Jack transmit packets to Steve on the interval $[0,T]$. We model the packet transmission times by a Poisson point process with parameter $\lambda$. Hence, the average number of packets that is sent from Jack to Steve in a duration of time $T$ is $\lambda T$.  Assume Alice wishes to use the channel to communicate with Bob without being detected by warden Willie, whom is observing the channel from Alice to Bob and attempting to detect any covert transmissions. 

In Scenario~1, which is analyzed in Section III, we assume that Willie does not have access to the contents of the packets and therefore is not able to perform authentication and determine whether a packet is from Jack or not.  We assume that Alice, with knowledge of $\lambda$, is allowed to use the channel only by inserting and transmitting her own packets to Bob, and Bob is able to authenticate, receive and remove the covert packets. Therefore Steve does not observe the covert packets.   We consider the rate $\Delta$ at which Alice can insert packets while remaining covert.   Willie knows that the legitimate communication is modeled by a Poisson point process with parameter $\lambda$, and he knows all of the characteristics of Alice's insertion scheme (rate, method of insertion, etc.) 


\begin{figure}
\begin{center}
\includegraphics[width=\textwidth/2,height=\textheight,keepaspectratio]{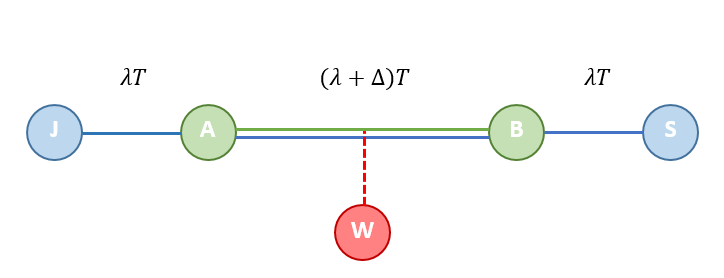}
\end{center}
 \caption{System configuration for Scenario~1: Alice adds her packets to the channel between Jack and Steve to communicate covertly with Bob. The average number of packets from Jack in a duration of time $T$ is $\lambda T$, and the average number of total packets observed by either Bob and Willie in $T$ is $\left(\lambda+ \Delta\right)T$. The blue color shows the legitimate communication and the green shows the covert communication.}
 \label{fig:SysMod}
 \end{figure}

In Scenario 2, which is analyzed in Section IV, we assume that Willie is able to access packet contents and hence authenticate whether a packet is coming from Jack.  Hence, Alice is not able to insert packets into the channel.  However, we allow Alice to buffer packets and release them when she desires into the channel, thereby enabling covert communication through packet timing control.  The problem is made non-trivial by assuming that Bob has access to the resulting packet stream only through an $M/M/1$ queue with service rate $\mu >  \lambda$; in other words, the service times of the packets are independently and identically distributed (i.i.d) random variables according to an exponential random variable with average $1/\mu$.   Again, Willie knows that the legitimate communication is modeled by a Poisson point process with parameter $\lambda$, and he knows all of the characteristics of Alice's packet buffering and release scheme except a secret key that is pre-shared between Alice and Bob. Fig.~\ref{fig:SysMod2} depicts the system model for this scenario.  In this scenario, we consider the rate at which Alice can reliably transmit information to Bob without detection by Willie.


 \begin{figure}
 \begin{center}
\includegraphics[width=\textwidth/2,height=\textheight,keepaspectratio]{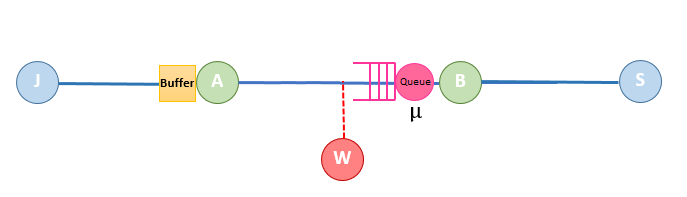}
 \end{center}
  \caption{System configuration for Scenario~2: 
Alice is able to buffer packets in order to alter packet timings.}
  \label{fig:SysMod2}
  \end{figure}

\subsection{Hypothesis Testing}



Willie is faced with a binary hypothesis test:  the null hypothesis ($H_0$) corresponds to the case that Alice is not transmitting, and the alternative hypothesis $H_1$ corresponds to the case that Alice is transmitting. We denote the distribution Willie observes under $H_1$ by $\mathbb{P}_{1}$, and the distribution Willie observes under $H_0$ by $\mathbb{P}_0$. Also, we denote by $\mathbb{P}_{FA}$ the probability of rejecting $H_0$ when it is true (type I error or false alarm), and $\mathbb{P}_{MD}$ the probability of rejecting $H_1$ when it is true (type II error or missed detection). We assume that Willie uses classical hypothesis testing with equal prior probabilities and seeks to minimize $\mathbb{P}_{FA} + \mathbb{P}_{MD}$; the generalization to arbitrary prior probabilities is straightforward~\cite{bash_jsac2013}.  

\subsection{Covertness}
Alice's transmission is {\em covert} if and only if she can lower bound Willie's sum of probabilities of error ($\mathbb{P}_{FA}+\mathbb{P}_{MD}$) by $1-\epsilon$ for any $\epsilon>0$~\cite{bash_jsac2013}.

\subsection{Reliability}

A transmission scheme is {\em reliable} if and only if the probability that a codeword transmission from Alice to Bob is unsuccessful is upper bounded by $\zeta$ for any $\zeta>0$.  Note that this metric applies only in Scenario 2.

\section{Covert Communication via Packet Insertion (Scenario 1)}
 
\label{theorem1}

In this section, we consider Scenario 1:  Willie cannot authenticate packets to see whether they are from Jack or Alice, and Alice is only allowed to send information to Bob by adding packets to the channel. 


\begin{thm}
Consider Scenario 1.  Alice can covertly insert $\mathcal{O}(\sqrt{\lambda T}) $ packets in the period of time $T$. Conversely, if Alice attempts to insert $\omega\left(\sqrt{\lambda T}\right)$ packets in the period of time $T$, there exists a detector that Willie can use to detect her with arbitrarily low sum of error probabilities $\mathbb{P}_{FA} + \mathbb{P}_{MD}$.

\end{thm}

\begin{proof}

{\it (Achievability)} 

\textbf{Construction}: Alice generates a Poisson point process with parameter $\Delta$ and based on this, adds her covert packets to the channel between Jack and Steve. Bob collects and removes the packets generated by Alice.


\textbf{Analysis}: First we show that the number of packets observed during the time period of length $T$ is a sufficient statistic by which Willie can perform the optimal hypothesis test to decide whether Alice is transmitting or not. The observation of the Poisson point process at Willie is characterized by the collection of inter-packet delays. Let $\bar{x}=\left(x_1,x_2,\cdots \right)$ be this observation vector, where $x_i$ is the delay between the arrival time of packet $i$ and packet $i+1$.  The optimal hypothesis test is based on a likelihood ratio (LRT) between the null hypothesis $H_0$ (Alice is not transmitting) and the alternative hypothesis $H_1$ (Alice is transmitting):
\begin{align}
\nonumber \Lambda(\bar{x}) &= \frac{f(\bar{x}|H_1)}{f(\bar{x}|H_0)}= \frac{\sum_{n_1} f(\bar{x} | H_1, n_1) \mathbb{P}_{1}(n_1) }{ \sum_{n_0} f(\bar{x} |  \nonumber  H_0, n_0) \mathbb{P}_{0}(n_0)}
\end{align}

\noindent where $\mathbb{P}_0$ is the probability distribution function for the number of packets that Willie observes under the null hypothesis $H_0$ (Alice is not transmitting) corresponding to a Poisson point process with rate $\lambda$, and $\mathbb{P}_1 $ is the probability distribution function for the number of packets that Willie observes under hypothesis $H_1$ (Alice is transmitting) corresponding to a Poisson point process with rate $\lambda+\Delta$.  Since the processes under $H_0$ and $H_1$ are Poisson, $f\left(\bar{x} \left|  H_0, n\right.\right) = f(\bar{x} |   H_1, n)$ for all $n$, and
\begin{align}
           \Lambda(\bar{x}) &= \frac{\sum_{n_1} f(\bar{x} | n_1) \mathbb{P}_{1}(n_1)}{  \sum_{n_0} f(\bar{x} | n_0) \mathbb{P}_{0}(n_0)}
\end{align}

\noindent  Suppose that $\bar{x}$ has $n$ non-zero entries for a given observation. Since $f(\bar{x} | n_1)=0$ when $n_1 \neq n$,
\begin{align}
\Lambda(\bar{x}) &= \frac{f(\bar{x} | n) \mathbb{P}_{1} (n)}  {f(\bar{x} | n) \mathbb{P}_{0}(n)}= \frac{\mathbb{P}_{1}(n)}{  \mathbb{P}_{0}(n)}
\end{align}

\noindent Therefore, the number of packets that Willie observes during time interval $T$ is a sufficient statistic for Willie's decision. 

Now, Willie applies the optimal hypothesis test on the number of packets during time $T$. Then \cite{bash_jsac2013}
\begin{align} \label{eq:th10e0001} \mathbb{P}_{FA}+\mathbb{P}_{MD} \geq 1-     \sqrt{\frac{1}{2} \mathcal{D}(\mathbb{P}_0 || \mathbb{P}_1)} \; 
\end{align}
\begin{align} \label{eq:th10e0002} \mathbb{P}_{FA}+\mathbb{P}_{MD} \geq 1-     \sqrt{\frac{1}{2} \mathcal{D}(\mathbb{P}_1 || \mathbb{P}_0)} \; 
\end{align}
\noindent where $\mathbb{E}[\cdot]$ denotes the expected value, and $\mathcal{D}(\mathbb{P}_0 || \mathbb{P}_1)$ is the relative entropy between $\mathbb{P}_0$ and $\mathbb{P}_1$. We next show how Alice can lower bound the sum of average error probabilities by upper bounding $  \sqrt{\frac{1}{2} \mathcal{D}(\mathbb{P}_0 || \mathbb{P}_1)}$. Consider the following fact (proved in the Appendix):

\begin{fact} \label{f1}
For two Poisson distributions $\mathbb{P}_{\lambda_1}(n)$, $\mathbb{P}_{\lambda_2}(n)$:
\begin{align}
\mathcal{D}(\mathbb{P}_{\lambda_1}(n)||\mathbb{P}_{\lambda_2}(n)) = \lambda_2-\lambda_1+\lambda_1 \log{\lambda_1/\lambda_2}
\end{align}
\end{fact}
By Fact.~\ref{f1}, for the given $ \mathbb{P}_1$ and $ \mathbb{P}_0$ the relative entropy is:
\begin{align}\label{eq:thm1d}
\mathcal{D}(\mathbb{P}_0 || \mathbb{P}_1) =  \Delta \cdot T - \lambda T \ln{\left( 1+\frac{\Delta}{\lambda}\right)}
\end{align}
\noindent We prove in the Appendix that:
\begin{align}\label{eq:ineq0}
\ln(1+u) \geq u-\frac{u^2}{2} \text{ for } u\geq 0,
\end{align}
\noindent Applying~\eqref{eq:ineq0} on~\eqref{eq:thm1d} yields:
\begin{equation}\label{eq:thm1d2}
\mathcal{D}(\mathbb{P}_0 || \mathbb{P}_1) \leq \frac{\Delta^2}{\lambda}T.
\end{equation}
\noindent Suppose Alice sets 
\begin{align} \label{eq:Delta}
\Delta = \sqrt{\frac{2\lambda}{T}}\epsilon
\end{align}
\noindent where $0<\epsilon<1$ is a constant. Then,
\begin{align}
\nonumber   \sqrt{\frac{1}{2} \mathcal{D}(\mathbb{P}_0 || \mathbb{P}_1)} \;   &\leq \epsilon
\end{align}
\noindent and $\mathbb{E}[\mathbb{P}_{FA}+\mathbb{P}_{MD}] \geq 1-\epsilon$ as long as $ \Delta \cdot T = \mathcal{O} \left(\sqrt{\lambda T}\right) $. 

{\it (Converse)} To establish the converse, we provide an explicit detector for Willie that is sufficient to limit Alice's throughput across all potential transmission schemes (i.e. not necessarily insertion according to a Poisson point process). Suppose that Willie observes a time interval of length $T$ and wishes to detect whether Alice is transmitting or not. Since he knows that the packet arrival process for the link between Jack and Steve is a Poisson point process with parameter $\lambda$, he knows the expected number of packets that he will observe. Therefore, he counts the number of packets $S$ in this duration of time and performs a hypothesis test by setting a threshold $U$ and comparing $S$ with $\lambda T+U$. If $S<\lambda T+U$, Willie accepts $H_0$; otherwise, he accepts $H_1$. Consider $\mathbb{P}_{FA}$
\begin{align}
\label{eq:PFAupperbound1} \mathbb{P}_{FA} = \mathbb{P}\left(S>\lambda T + U | H_0 \right)=\mathbb{P}\left(S-\lambda T > U | H_0 \right)
\leq\mathbb{P}\left(|S-\lambda T| > U | H_0 \right)
\end{align}
\noindent When $H_0$ is true, Willie observes a Poisson point process with parameter $\lambda$; hence,
\begin{align}
\label{eq:th1con1}\mathbb{E}\left[S\left|H_0\right.\right]&=\lambda T\\
\label{eq:th1con2}\mathrm{Var}\left[S\left|H_0\right.\right]&={\lambda T}
\end{align}
\noindent  Therefore, applying Chebyshev's inequality on~\eqref{eq:PFAupperbound1} yields $\mathbb{P}_{FA} \leq \frac{\lambda T}{U^2 }$. Thus, Willie can achieve $\nonumber \mathbb{P}_{FA} \leq \alpha$ for any $0<\alpha<1$ if he sets $U=\sqrt{\frac{\lambda T}{ \alpha}}$.
\noindent Next, we will show that if Alice transmits $ \Delta \cdot T=\omega\left(\sqrt{\lambda T}\right)$ covert packets, she will be detected by Willie with high probability. Consider $\mathbb{P}_{MD}$
\begin{align}
 \nonumber \mathbb{P}_{MD}=\mathbb{P}\left(S\leq \lambda T + U | H_1\right) &=\mathbb{P}\left(S - \left(\lambda + \Delta \right) T \leq  U - \Delta \cdot T | H_1\right)\\
\nonumber &=\mathbb{P}\left(- \left(S - \left(\lambda + \Delta \right) T\right) \geq  -\left(U - \Delta \cdot T\right) | H_1\right)\\
\label{eq:PMDupperbound1} &\leq \mathbb{P}\left( \left|S - \left(\lambda + \Delta \right) T\right| \geq  \left|U - \Delta \cdot T\right| | H_1\right)
\end{align}
\noindent When $H_1$ is true, since Willies observes a Poisson point process with parameter $\lambda+\Delta$,
\begin{align}
\label{eq:th1con3}\mathbb{E}\left[S\left|H_1\right.\right]&=\left(\lambda + \Delta\right) T\\
\label{eq:th1con4}\mathrm{Var}\left[S\left|H_1\right.\right]&=\left(\lambda + \Delta\right) T
\end{align}
\noindent Therefore, applying Chebyshev's inequality on~\eqref{eq:PMDupperbound1} yields $ \mathbb{P}_{MD}\leq \frac{\left(\lambda + \Delta\right) T}{\left( U- \Delta \cdot T \right)^2}$. Since $U=\sqrt{\frac{\lambda T}{ \alpha}}$, if Alice sets $ \Delta \cdot T=\omega\left(\sqrt{\lambda T}\right)$, Willie can achieve $\mathbb{P}_{MD} < \beta$ for any $0<\beta<1$. Combined with the results for probability of false alarm above, if Alice sets $\Delta.T=\omega\left(\sqrt{\lambda T}\right)$, Willie can choose a $U=\sqrt{\frac{\lambda T}{ \alpha}}$ to achieve any (small) $\alpha>0$ and $\beta>0$ desired.
\end{proof}


\section{Covert Timing Channel Using Inter-Packet Delays (Scenario 2)}
\label{theorem2}
In this section, we consider Scenario 2:  Willie can authenticate packets to determine whether or not they were generated by the legitimate transmitter Jack. Therefore, Alice cannot insert packets into the channel; rather, we assume that Alice is able to buffer packets and release them when she desires; hence, she can encode information in the inter-packet delays by using a secret codebook shared with Bob. 

In the construction below, each of Alice's codewords will consist of a sequence of inter-packet delays to be employed to convey the corresponding message.  But this immediately presents a problem for Alice, as she must have a packet in her buffer to transmit at the appropriate time.  To clarify this issue, suppose Alice were to attempt to start codeword transmission immediately.  Suppose that the first letter of Alice's codeword for the current message is $10\mu s$, and thus she has to apply this delay between the first and second packet. To send this letter, Alice receives the first packet from Jack and immediately transmits it to Bob. Then, she waits for the second packet. If the second packet arrives earlier than $10\mu s$, she stores the packet and transmits it $10\mu s$ after the first transmission. But, if the second packet arrives from Jack after the $10\mu s$ interval, she cannot accurately transmit the codeword.  Hence, the codeword failure rate would be quite high with such a scheme.

This suggests that Alice must build up some number of packets in her buffer before starting codeword transmission, and, in fact, this is the scheme that we employ.  In particular, Alice will employ a two-phase system.  In the first phase, she will (slightly) slow down the transmission of packets from Jack to Steve so as to build up a backlog of packets in her buffer.  Then, during the codeword transmission phase, she will release packets from her buffer with the inter-packet delays prescribed by the codeword corresponding to the message, while continuing to buffer arriving packets from Jack.  Two questions dictate the resulting throughput of the scheme:  (1)  Phase I: how much can Alice slow down the packet stream from Jack to Steve without it being detected by warden Willie?; (2) Phase II:  how many packets must Alice accumulate in her buffer by the start of Phase II so as to, with high probability, have a packet in her buffer at all of the times required by the codeword?
To answer the first question, we calculate the number of packets that Alice can covertly collect by slowing down the packet stream from Jack to Steve in a duration of time $T$ and present the results in Lemma~1. Then, in Theorem~2, we answer the second question to present an achievability result for the number of bits that Alice can reliably and covertly transmit to Bob through packet timings. 

\begin{lem}
If Alice can buffer packets on the link from Jack to Steve, she can covertly buffer $ \mathcal{O}(\sqrt{\lambda T}) $ packets in a period of time of length $T$. Conversely, if Alice buffers $\omega\left(\sqrt{\lambda T}\right) $ packets in the period of time $T$, there exists a detector that Willie can use to detect such buffering with arbitrarily low sum of error probabilities $\mathbb{P}_{FA} + \mathbb{P}_{MD}$.

\end{lem}

\begin{proof}
For the achievability result, Alice buffers packets by slowing down the incoming process of rate $\lambda$ to one of rate $\lambda-\Delta$ (i.e. slightly scaling up the packet timings) and buffering the packets of the (faster: of rate $\lambda$) incoming process that have yet to exit the buffer by the (slower: of rate $\lambda-\Delta$) outgoing process.
$\mathbb{P}_0$ is the probability distribution function for Willie's observation under the null hypothesis $H_0$ and corresponds to a Poisson point process with rate $\lambda$, and $\mathbb{P}_1 $ is the probability distribution function for observations under the hypothesis $H_1$ and corresponds to a Poisson point process with rate $\lambda-\Delta$. Similar to~\eqref{eq:thm1d}, by Fact~\ref{f1}: 
\begin{align}\label{eq:thm1d3}
 \nonumber \mathcal{D}(\mathbb{P}_1 || \mathbb{P}_0) &=  \Delta \cdot T - (\lambda - \Delta)T \ln{\left( 1+\frac{\Delta}{\lambda-\Delta}\right)}\\
\nonumber &\stackrel{(a)}{\leq}   \Delta \cdot T - (\lambda - \Delta)T \left(\frac{\Delta}{\lambda-\Delta}-\frac{1}{2}\left(\frac{\Delta}{\lambda-\Delta}\right)^2\right)\\
 &=   \Delta \cdot T - T \left({\Delta}-\frac{\Delta^2}{2(\lambda-\Delta)}\right)=  \frac{T \Delta^2}{2(\lambda-\Delta)}
\end{align}
\noindent where $(a)$ follows from inequality~\eqref{eq:ineq0}. If $ \Delta = \epsilon \sqrt{2\lambda/}$ , then $\mathcal{D}(\mathbb{P}_1 || \mathbb{P}_0)\leq 2 \epsilon^2$ as $T \to \infty$. Therefore, by~\eqref{eq:th10e0002} and~\eqref{eq:thm1d3}, Alice can achieve $\mathbb{P}_{FA}+\mathbb{P}_{MD} \geq 1-\epsilon$ as for any $0<\epsilon<1$ long as $ \Delta \cdot T= \epsilon \sqrt{2\lambda T} = \mathcal{O} \left(\sqrt{\lambda T}\right) $. The proof of the converse follows that of Theorem 1 with minor modifications.
\end{proof}
\begin{thm}
Consider Scenario 2.  By embedding information in the inter-packet delays, Alice can can covertly and reliably transmit $\mathcal{O}\left(\lambda T \right)$ bits to Bob in a time interval of length $T$. 
\end{thm}

\begin{proof}

\textbf{Construction}: To establish covert communication over the timing channel, Alice and Bob share a secret key (codebook) to which Willie does not have access. To build a codebook, a set of $M$ independently generated codewords $\{C(W_l)\}_{l=1}^{l=M}$ are generated for messages $\{W_l\}_{l=1}^{l=M}$ according to realizations of a Poisson process with parameter $\lambda$ that mimics the overt traffic on the channel between Jack and Steve, where $M$ is defined later.  In particular, to generate a codeword $C(W_l)$, first a random variable $N$ is generated according to a Poisson distribution with mean $\lambda T$. Then, $N$ inter-packet delays are generated by placing $N$ points uniformly and independently on an interval of length $T$ \cite{verdubitsq} (see  Fig.~\ref{fig:codebook}). For each message transmission, Alice uses a new codebook to encode the message into a codeword. According to the codebook, each message corresponds to a codeword that is a series of inter-packet delays. To send a codeword, Alice applies the inter-packet delays to the packets that are being transmitted from Jack to Steve. 

\begin{figure}
\begin{center}
\includegraphics[width=\textwidth/2,height=\textheight,keepaspectratio]{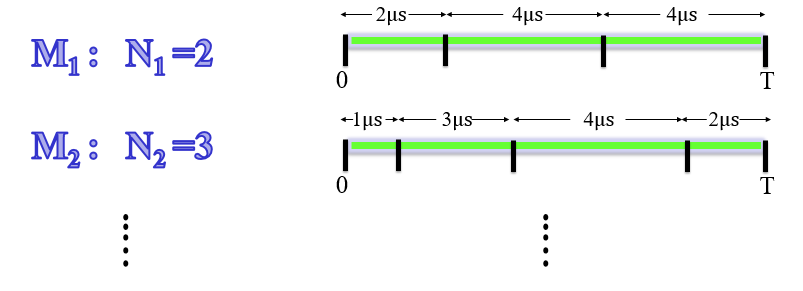}
\end{center}
 \caption{Codebook generation: Alice and Bob share a codebook (secret), which specifies the sequence of inter-packet delays corresponding to each message.  To generate each codeword, a number $N$ is generated according to the Poisson distribution with parameter $\lambda T$, and then $N$ points are placed uniformly and randomly on the time interval $[0,T]$.}
 \label{fig:codebook}
 \end{figure}

Per above, Alice's communication includes two phases: a buffering phase and a transmission phase. During the buffering phase of length $ \psi T$, where $0< \psi<1$ is a parameter to be defined later, Alice slows down the packet transmission process to $\lambda-\Delta$ in order to build up packets in her buffer.  In particular, Alice's purpose in the first phase is to buffer enough packets to ensure she will not run out of packets during the transmission phase, of length $T'=(1- \psi)T$, with high probability (see Fig.~\ref{fig:Twophased}).

\begin{figure}
\begin{center}
\includegraphics[width=\textwidth/2,height=\textheight,keepaspectratio]{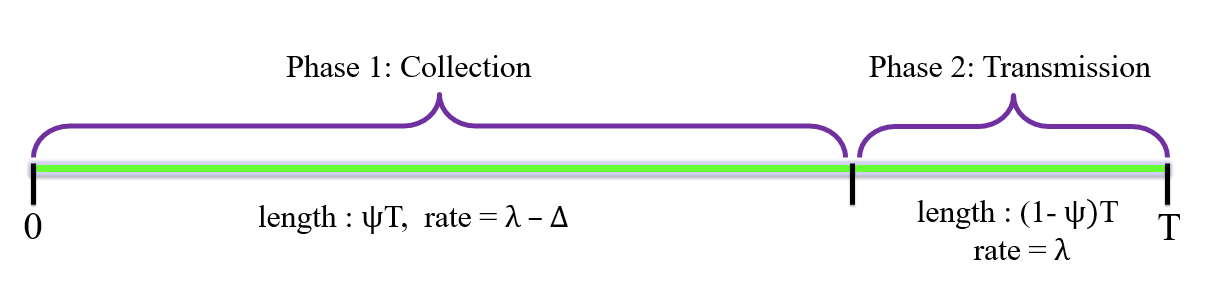}
\end{center}
 \caption{Two-phased construction: Alice divides the duration of time $T$ into two phases with lengths $ \psi T$ and $(1- \psi)T$. In the first phase, Alice slows down the transmission to $\lambda-\Delta$ and buffers the excess packets. In the next phase, she transmits packets to Bob at rate $\lambda$ according to the inter-packet delays in the codeword corresponding to the message to be transmitted.}
 \label{fig:Twophased}
 \end{figure}


\textbf{Analysis}: ({\em Covertness})  Suppose that Willie knows when each of the two phases start and end. According to Lemma~1, during the first phase, Alice can lower bound the sum of Willie's error probability by $1-\epsilon$ and collect $m=\Delta \cdot \psi T$ packets. Thus, in this phase, Alice's buffering is covert. By \eqref{eq:Delta},
\begin{align}
\label{eq:th2m} m= \epsilon \sqrt{2\lambda  \psi T}
\end{align}
During the second phase, the packet timings corresponding to the selected codeword are an instantiation of a Poisson point process with rate $\lambda$ and hence the traffic pattern is indistinguishable from the pattern that Willie expects on the link from Jack to Steve.  Hence, the scheme is covert.

({\em Reliability})  Next we show that Alice will have a reliable communication to Bob.  The notion of reliability is tied to two events. First, Bob should be able to decode the message with arbitrarily low probability of error, which follows by adapting the coding scheme of  \cite{verdubitsq}.   Second, Alice needs to avoid a ``failure'' event, in which Alice is unable to create the packet timings for the selected codeword because she has run out of packets in her buffer at some point in the codeword transmission process.   Therefore, in the transmission phase, we need to demonstrate that Alice can achieve $\mathbb{P}_f<\zeta$ for any $\zeta>0$, where $\mathbb{P}_f$ is probability of the event ``failure''.  

 Suppose that Alice has accumulated $m$ packets and wishes to communicate with Bob during an interval of length $T'$. There are two processes which impact the probability of running out of packets in Alice's buffer: the transmission of packets to Bob and the reception of packets from Jack (that are intended for Steve). Each of these processes is a Poisson point process with parameter $\lambda$. Consider the union of these two processes as a single process, termed ``T-R Process''.   The T-R Process is a Poisson point process with parameter $2\lambda$. Note that the probability of an event in this aggregate process corresponding to Alice buffering a packet is the same as the probability that it corresponds to her sending a packet to convey the next code symbol. Therefore, we toss a fair coin for each packet; if it is heads, Alice buffers a received packet; otherwise, she releases a packet from her buffer.

Let $K$ be the total number of received and transmitted packets during interval $T'$. 
\begin{align}
 \mathbb{P}_f&=  \mathbb{P}\left(\mathcal{F}\cap \{K> 4 \lambda T'\}\right) +\mathbb{P}\left(\mathcal{F}\cap \{K\leq 4 \label{eq:th2.1}\lambda T'\}\right) \\
\label{eq:th202} & \leq  \mathbb{P}\left(K\geq 4 \lambda T'\right) +\mathbb{P}\left(\mathcal{F}\cap \{K\leq 4 \lambda T'\}\right)
\end{align}
\noindent where $\mathcal{F}$ is the ``failure'' event, and~\eqref{eq:th2.1} is due to the law of total probability. Consider $\mathbb{P}\left(K> 4 \lambda T'\right)$. Using the Chernoff bound for the T-R Poisson point process yields:
\begin{align}
\mathbb{P}\left(K\geq 4 \lambda T'\right) &\leq \frac{e^{-2\lambda T' }\left( 2 \lambda T' e\right)^{4 \lambda T'}}{\left(4 \lambda T'\right)^{4 \lambda T'}}=\left(\frac{e}{4}\right)^{2\lambda T'}
\end{align}

\noindent Thus, for any $0<\zeta<1$, if $\lambda T'$ is large enough, 
\begin{align}
\label{th2.r1} \mathbb{P}\left(K\geq 4 \lambda T'\right)<\frac{\zeta}{2}
\end{align}
Now, consider $ \mathbb{P}\left(\mathcal{F}\cap \{K\leq 4 \lambda T'\}\right)$

\begin{align}
\label{eq:th2.e1} \mathbb{P}\left(\mathcal{F}\cap \{K\leq 4 \lambda T'\}\right)&=\mathbb{P}\left(\mathcal{F}\left| K\leq 4 \lambda T'\right.\right)\mathbb{P}\left(K\leq 4 \lambda T'\right)\\
\label{eq:th2.e2}&\leq\mathbb{P}\left(\mathcal{F}\left| K\leq 4 \lambda T'\right.\right)\\
\label{eq:th2.e21}&\leq   \mathbb{P}\left(\mathcal{F}|K=k'\right)|_{k'=4\lambda T'}
\end{align}

\noindent where $\mathbb{P}\left(\mathcal{F}|K=k'\right)$ is the probability of ``failure'' given the total number of received and transmitted packets in $T'$ is $k'$, and~\eqref{eq:th2.e21} follows from~\eqref{eq:th2.e2} because  $\mathbb{P}\left(\mathcal{F}|K=k'\right)$ is a monotonically increasing function of $k'$. This is true since with $m$ initial packets, if the total number of received and transmitted packets increases, the probability of ``failure'' increases.

To calculate $\mathbb{P}\left(\mathcal{F}|K=k'\right)$, we model the T-R Process by a one-dimensional random walk. Suppose that there is a random walker on the $z$ axis that is initially located at the origin $z=0$, and walks at each step from $z$ to either $z+1$ or $z-1$ with equal probability. If Alice transmits a packet, the walker's location is increased by $1$, and if Alice receives a packet, the walker's location is decrease by $1$. Given Alice has $m$ packets before starting this process, the event of ``failure'' in the T-R Process is the same as the event that, at some point, the random walker (that has started from origin $z=0$) hits state $z=m+1$.  Clearly, this can only happen if the walker steps from state $z=m$ to state $z=m+1$ at some point. Since the number of times that Alice either receives or transmits a packets in a duration of time $T'$ is the same as the number of steps of the walker $K$:
\begin{align}
\label{eq:th2e2.5}\mathbb{P}\left(\mathcal{F}|K=k'\right) &= 1 - \mathbb{P}_m(k')
\end{align}
\noindent where $\mathbb{P}_m\left(k'\right)$ is the probability that the walker's location remains $x\leq m$ during the first $k'$ steps. From (Eq. (12.13) in \cite{gut2009stopped}):
\begin{align}
\label{eq:survp} \lim\limits_{k'\to \infty} \mathbb{P}_m\left(k'\right)=\operatorname{erf}\left(\frac{m  }{\sqrt{2 k' }}\right)
\end{align}
\noindent where $\operatorname{erf}\left(\cdot\right)$ is the error function. By~\eqref{eq:th2m}, 
\begin{align}
\label{eq:th2e3} \lim\limits_{T\to \infty} \mathbb{P}_m(k')|_{k'=4\lambda T'} =\operatorname{erf}\left(\frac{m }{\sqrt{8\lambda T' }}\right)
  &=\operatorname{erf}\left(\frac{\epsilon \sqrt{2 \lambda  \psi T} }{ \sqrt{8 \lambda  (1- \psi)T}}\right)\\
\label{eq:th2e4}  &=\operatorname{erf}\left(\frac{\epsilon  }{2 }\sqrt{\frac{ \psi}{1- \psi}}\right)
\end{align}
\noindent By~\eqref{eq:th2.e21},~\eqref{eq:th2e2.5},~\eqref{eq:th2e4}
\begin{align}
\nonumber  \lim\limits_{T\to \infty} \mathbb{P}\left(\mathcal{F}\cap \{K\leq 4 \lambda T'\}\right)\leq   \lim\limits_{T\to \infty} \mathbb{P}\left(\mathcal{F}|K=4\lambda T'\right)=\left.1 -  \lim\limits_{T\to \infty} \mathbb{P}_m\left(k'\right)\right|_{k'=4\lambda T'}1 - \operatorname{erf}\left(\frac{\epsilon  }{2 }\sqrt{\frac{ \psi}{1- \psi}}\right).
\end{align}
\noindent Therefore, if 
\begin{align}
\label{eq:th2rho} \frac{ \psi}{1- \psi} =\left( \frac{2 }{\epsilon} \operatorname{erf}^{-1} \left(1-\frac{\zeta}{2 }\right)\right)^2 
\end{align}
where $\operatorname{erf}^{-1}\left(\cdot\right)$ is the inverse error function, 
\begin{align}
\label{eq:th2.r2}  \lim\limits_{T\to \infty} \mathbb{P}\left(\mathcal{F}\left|\{K\leq 4 \lambda T'\}\right.\right) \leq \frac{\zeta}{2}.
\end{align}

\noindent Therefore, by~\eqref{th2.r1},~\eqref{eq:th2.r2} Alice can achieve $\mathbb{P}_f < \zeta$ for any $0<\zeta<1$  and thus the transmission is reliable. 

({\em Number of Covert Bits}) Alice's rate of packet transmission in the second phase is $\lambda$ and the capacity of the queue for conveying information through inter-packet delays is \cite[Theorems 4 and 6]{verdubitsq}.
\begin{align}
\label{th2:cap}C(\lambda)=\lambda \log{\frac{\mu}{\lambda}} \text{  nats/sec}
\end{align}
\noindent Since the length of the second phase is $T'=T(1-\psi)$, the amount of information that Alice sends through inter-packet delays is $C(\lambda) T (1-\psi) = \mathcal{O} \left(\lambda T\right)$.

({\em Size of the Codebook}) Finally, to complete the construction, we find the number $M$ of codewords in the codebook. According to \cite[Definition 1]{verdubitsq}, the rate of the codebook is $\frac{\log M}{T'} = \frac{\log M}{T (1-\psi)}$. Since the capacity of the queue at output rate $\lambda$ is the maximum achievable rate at output rate $\lambda$ (see \cite[Definition 2]{verdubitsq}), by~\eqref{th2:cap}, the size of the codebook is
\begin{align}
M = e^{{\left(1- \psi\right) \lambda T \log{(\mu/\lambda)} }}
\end{align}
\noindent where $1-\psi = \left(\left( \frac{2 }{\epsilon} \operatorname{erf}^{-1} \left(1-{\zeta}\right)\right)^2 +1\right)^{-1}$.
 
\end{proof}
\section{Conclusion}
In this paper, we have considered two scenarios of covert communication on a timing channel.  In the first one, Alice is able to insert her own packets onto the channel but not modify the timing of other packets.  We established that, if the packet transmission between legitimate nodes is modeled as a Poisson point process with parameter $\lambda$, Alice can transmit $\mathcal{O}\left(\sqrt{\lambda T}\right)$ packets to Bob in a period of duration $T$ without being detected by warden Willie. In the converse, we showed that if Alice inserts more than $\mathcal{O}\left(\sqrt{\lambda T}\right)$ packets, she will be detected by Willie. Next, we analyzed the scenario where Alice cannot insert packets but instead is able to buffer packets and release them onto the channel at a later time.  We showed that if Alice waits $ \psi T$ before codeword transmission, where $ \psi = \Theta\left(1\right)$, and buffers packets, she is able to transmit $\mathcal{O}\left(\lambda T\right)$ bits to Bob reliably and covertly in time period $T$, given Alice and Bob share a secret of sufficient length (codebook). 
\bibliographystyle{ieeetr}

\begin{thebibliography}{10}

\bibitem{soltani2016allertonarxiv}
R.~Soltani, D.~Goeckel, D.~Towsley, and A.~Houmansadr, ``Covert communications
  on renewal packet channels.'' \url{https://arxiv.org/abs/1610.00368}, 2016.

\bibitem{talb2006}
J.~Talbot and D.~Welsh, {\em Complexity and Cryptography: An Introduction}.
\newblock Cambridge University Press, 2006.

\bibitem{bloch11pls}
M.~Bloch and J.~Barros, {\em Physical-Layer Security}.
\newblock Cambridge, UK: Cambridge University Press, 2011.

\bibitem{ker07pool}
A.~D. Ker, ``Batch steganography and pooled steganalysis,'' vol.~4437 of {\em
  Lecture Notes in Computer Science}, pp.~265--281, Springer Berlin Heidelberg,
  2007.

\bibitem{bash12sqrtlawisit}
B.~A. Bash, D.~Goeckel, and D.~Towsley, ``Square root law for communication
  with low probability of detection on {AWGN} channels,'' in {\em Proc. {IEEE}
  Int. Symp. Inform. Theory (ISIT)}, (Cambridge, MA, USA), pp.~448--452, July
  2012.

\bibitem{bash_jsac2013}
B.~Bash, D.~Goeckel, and D.~Towsley, ``Limits of reliable communication with
  low probability of detection on {AWGN} channels,'' {\em Selected Areas in
  Communications, IEEE Journal on}, vol.~31, pp.~1921--1930, September 2013.

\bibitem{boulat_commmagg}
B.~A. Bash, D.~Goeckel, D.~Towsley, and S.~Guha, ``Hiding information in noise:
  Fundamental limits of covert wireless communication,'' {\em IEEE
  Communications Magazine: Special Issue on Wireless Physical Layer Security},
  Dec. 2015.
\newblock to appear.

\bibitem{commL}
D.~Goeckel, B.~Bash, S.~Guha, and D.~Towsley, ``Covert communications when the
  warden does not know the background noise power,'' {\em IEEE Communication
  Letters}, Aug 2015.
\newblock submitted.

\bibitem{bash13quantumlpdisit}
B.~A. Bash, S.~Guha, D.~Goeckel, and D.~Towsley, ``{Quantum Noise Limited
  Communication with Low Probability of Detection},'' in {\em Proc. {IEEE} Int.
  Symp. Inform. Theory (ISIT)}, (Istanbul, Turkey), pp.~1715--1719, July 2013.

\bibitem{bash14timing}
B.~Bash, D.~Goeckel, and D.~Towsley, ``{LPD communication when the warden does
  not know when},'' in {\em Information Theory (ISIT), 2014 IEEE International
  Symposium on}, pp.~606--610, 2014.

\bibitem{soltani14netlpdallerton}
R.~Soltani, B.~A. Bash, D.~Goeckel, S.~Guha, and D.~Towsley, ``Covert
  single-hop communication in a wireless network with distributed artificial
  noise generation,'' in {\em Proc. Conf.~on Commun., Control,
  Comp.~(Allerton)}, (Monticello, IL, USA), pp.~1078--1085, 2014.

\bibitem{tammy_asilomar2015}
T.~Sobers, D.~Goeckel, B.~Bash, S.~Guha, and D.~Towsley, ``Covert communication
  with the help of an uninformed jammer achieves positive rate,'' in {\em Proc.
  Asilomar Conf. on Signals, Systems, and Computers}, (Monterey, CA, USA),
  2015.

\bibitem{che13sqrtlawbscisit}
P.~H. Che, M.~Bakshi, and S.~Jaggi, ``Reliable deniable communication: Hiding
  messages in noise,'' in {\em Proc. {IEEE} Int. Symp. Inform. Theory (ISIT)},
  (Istanbul, Turkey), pp.~2945--2949, July 2013.

\bibitem{kadhe14sqrtlawmultipathisit}
S.~Kadhe, S.~Jaggi, M.~Bakshi, and A.~Sprintson, ``Reliable, deniable, and
  hidable communication over multipath networks,'' in {\em Proc. {IEEE} Int.
  Symp. Inform. Theory (ISIT)}, (Honolulu, HI, USA), pp.~611--615, July 2014.

\bibitem{hou14isit}
J.~Hou and G.~Kramer, ``Effective secrecy: Reliability, confusion and
  stealth,'' in {\em Proc. {IEEE} Int. Symp. Inform. Theory (ISIT)}, (Honolulu,
  HI, USA), pp.~601--605, July 2014.

\bibitem{bloch15covert-isit}
M.~Bloch, ``Covert communication over noisy memoryless channels: A
  resolvability perspective,'' in {\em Proc. {IEEE} Int. Symp. Inform. Theory
  (ISIT)}, (Hong Kong, China), 2015.

\bibitem{wang15covert-isit}
L.~Wang, G.~W. Wornell, and L.~Zhang, ``Limits of low-probability-of-detection
  communication over a discrete memoryless channel,'' in {\em Proc. {IEEE} Int.
  Symp. Inform. Theory (ISIT)}, (Hong Kong, China), 2015.

\bibitem{lee14posratecovert}
S.~Lee and R.~Baxley, ``{Achieving positive rate with undetectable
  communication over AWGN and Rayleigh channels},'' in {\em {Proc. {IEEE}
  Int.~Conf.~Commun.~(ICC)}}, pp.~780--785, June 2014.

\bibitem{jaggi_uncertain}
P.~H. Che, M.~Bakshi, C.~Chan, and S.~Jaggi, ``Reliable deniable communication
  with channel uncertainty,'' in {\em Proc. Information Theory Workshop (ITW)},
  2014.

\bibitem{Girling87}
C.~Girling, ``Covert channels in lan's,'' {\em Software Engineering, IEEE
  Transactions on}, vol.~SE-13, pp.~292--296, Feb 1987.

\bibitem{berk2005detection}
V.~Berk, A.~Giani, and G.~Cybenko, ``Detection of covert channel encoding in
  network packet delays (tech. rep. tr2005-536),'' {\em Department of Computer
  Science, Dartmouth College (November 2005)}, 2005.

\bibitem{shah2006keyboards}
G.~Shah, A.~Molina, M.~Blaze, {\em et~al.}, ``Keyboards and covert channels.,''
  in {\em USENIX Security}, 2006.

\bibitem{liu2009hide}
Y.~Liu, D.~Ghosal, F.~Armknecht, A.-R. Sadeghi, S.~Schulz, and
  S.~Katzenbeisser, ``{Hide and seek in time--robust covert timing channels},''
  in {\em Computer Security--ESORICS 2009}, pp.~120--135, Springer, 2009.

\bibitem{houmansadr2011coco}
A.~Houmansadr and N.~Borisov, ``{CoCo: coding-based covert timing channels for
  network flows},'' in {\em Information Hiding}, pp.~314--328, Springer, 2011.

\bibitem{verdubitsq}
V.~Anantharam and S.~Verdu, ``Bits through queues,'' {\em Information Theory,
  IEEE Transactions on}, vol.~42, no.~1, pp.~4--18, 1996.

\bibitem{dunn2009}
B.~P. Dunn, M.~Bloch, and J.~N. Laneman, ``Secure bits through queues,'' in
  {\em Networking and Information Theory, 2009. ITW 2009. IEEE Information
  Theory Workshop on}, pp.~37--41, IEEE, 2009.

\bibitem{millen1989finite}
J.~K. Millen, ``Finite-state noiseless covert channels,'' in {\em Computer
  Security Foundations Workshop II, 1989., Proceedings of the}, pp.~81--86,
  IEEE, 1989.

\bibitem{Anand1998inf}
A.~S. Bedekar and M.~Azizoglu, ``The information-theoretic capacity of
  discrete-time queues,'' {\em Information Theory, IEEE Transactions on},
  vol.~44, no.~2, pp.~446--461, 1998.

\bibitem{sekke2007capacity}
S.~H. Sellke, C.-C. Wang, N.~Shroff, and S.~Bagchi, ``Capacity bounds on timing
  channels with bounded service times,'' in {\em Information Theory, 2007. ISIT
  2007. IEEE International Symposium on}, pp.~981--985, IEEE, 2007.

\bibitem{sekke2009}
S.~Sellke, C.-C. Wang, S.~Bagchi, and N.~Shroff, ``Tcp/ip timing channels:
  Theory to implementation,'' in {\em INFOCOM 2009, IEEE}, pp.~2204--2212,
  April 2009.

\bibitem{Mosko92Capac}
I.~Moskowitz and A.~Miller, ``The channel capacity of a certain noisy timing
  channel,'' {\em Information Theory, IEEE Transactions on}, vol.~38,
  pp.~1339--1344, Jul 1992.

\bibitem{Moskowitz1994}
I.~Moskowitz and A.~Miller, ``Simple timing channels,'' in {\em Research in
  Security and Privacy, 1994. Proceedings., 1994 IEEE Computer Society
  Symposium on}, pp.~56--64, May 1994.

\bibitem{kiyavash2009J}
N.~Kiyavash, T.~Coleman, and M.~Rodrigues, ``Novel shaping and
  complexity-reduction techniques for approaching capacity over queuing timing
  channels,'' in {\em Communications, 2009. ICC '09. IEEE International
  Conference on}, pp.~1--5, June 2009.

\bibitem{kiyavash2009A}
N.~Kiyavash and T.~Coleman, ``Covert timing channels codes for communication
  over interactive traffic,'' in {\em Acoustics, Speech and Signal Processing,
  2009. ICASSP 2009. IEEE International Conference on}, pp.~1485--1488, April
  2009.

\bibitem{Arch2012}
R.~Archibald and D.~Ghosal, ``A covert timing channel based on fountain
  codes,'' in {\em Trust, Security and Privacy in Computing and Communications
  (TrustCom), 2012 IEEE 11th International Conference on}, pp.~970--977, June
  2012.

\bibitem{gianvecchio2007detecting}
S.~Gianvecchio and H.~Wang, ``Detecting covert timing channels: an
  entropy-based approach,'' in {\em Proceedings of the 14th ACM conference on
  Computer and communications security}, pp.~307--316, ACM, 2007.

\bibitem{gianvecchio2008model}
S.~Gianvecchio, H.~Wang, D.~Wijesekera, and S.~Jajodia, ``Model-based covert
  timing channels: Automated modeling and evasion,'' in {\em Recent Advances in
  Intrusion Detection}, pp.~211--230, Springer, 2008.

\bibitem{gut2009stopped}
A.~Gut, {\em Stopped random walks}.
\newblock Springer, 1998.

\end{thebibliography}

\appendix
\paragraph{Proof of fact~\ref{f1}}
\begin{align}
\nonumber \mathcal{D}(\mathbb{P}_{\lambda_1}(n)||\mathbb{P}_{\lambda_2}(n)) = \sum_{n=0}^{\infty} \mathbb{P}_{\lambda_1}(n) \log(\mathbb{P}_{\lambda_1}(n))/\mathbb{P}_{\lambda_2}(n))
&= \sum_{n} \mathbb{P}_{\lambda_1}(n) \log{\frac{\lambda_1^n e^{-\lambda_1}/n!}{\lambda_2^n e^{-\lambda_2}/n!}}\\
\nonumber &= \sum_{n} \mathbb{P}_{\lambda_1}(n) \left(\log{\frac{\lambda_1^n }{\lambda_2^n}}  + {\lambda_2-\lambda_1}\right)\\
\nonumber &= (\lambda_2-\lambda_1) + \sum_{n} \mathbb{P}_{\lambda_1}(n) n \log{\frac{\lambda_1 }{\lambda_2}}  \\
\nonumber &= (\lambda_2-\lambda_1)  + \log{\frac{\lambda_1 }{\lambda_2}}  \mathbb{E}_1[N] \\
\nonumber &= (\lambda_2-\lambda_1)  +\lambda_1 \log{\frac{\lambda_1 }{\lambda_2}} 
\end{align}
\noindent where $\mathbb{E}_1[K] $ is the expected value of the Poisson distribution with parameter $\lambda_1$.

\paragraph{Proof of~\eqref{eq:ineq0}}

Consider $x\geq 0$, $f(x)=\ln(1+x)$, and $g(x)=x-\frac{x^2}{2}$. Therefore
\begin{align}
f'(x)-g'(x)=  \frac{1}{1+x} -\left(1-x\right) =\frac{x^2}{1+x} \geq 0  \end{align}
\noindent On the other hand $f(0)=g(0)=0$, therefore
\begin{align}
f(x)-g(x)=\int\limits_{0}^{x} \left(f'\left(x\right)-g'\left(x\right)\right) dx  \geq 0
\end{align}
\noindent Thus, $\ln(1+x) \geq x-\frac{x^2}{2}$ for $x\geq 0$.

\end{document}